\documentclass[letterpaper, 10 pt, conference]{ieeeconf}  
\IEEEoverridecommandlockouts 

\usepackage{cite}

\usepackage{xcolor}
\usepackage{epstopdf}
\usepackage{mathtools}
\usepackage{fancyhdr}
\usepackage{subcaption}
\usepackage{amsfonts}
\usepackage{amssymb}
\usepackage[T1]{fontenc}
\usepackage{color}	
\usepackage{siunitx}
\usepackage{multirow}
\usepackage{placeins}
\usepackage{booktabs}
\usepackage{rotating}
\usepackage{array}
\usepackage{physics}
\usepackage{dsfont}
\usepackage{bbm}
\usepackage{bbold}
\usepackage{caption}
\usepackage{subcaption}
\usepackage[hyphens]{url}

\usepackage{epsfig}
\usepackage{epstopdf}

\usepackage{xcolor}

\newcommand{\B}[1]{{\color{blue}#1}}

\newcommand{\mc}[1]{\mathcal{#1}}
\newcommand{\R}{\mathbb{R}}

\newcommand{\bs}[1]{\boldsymbol{#1}}

\newcommand{\tup}[1]{\textup{#1}}
\DeclareMathOperator*{\gph}{gph}   
\DeclareMathOperator*{\argmin}{argmin}   
  
\usepackage{mathtools}
\usepackage[acronym]{glossaries}
\newacronym{DR}{DR}{Demand Response}
\newacronym{DSO}{DSO}{Distribution System Operator}
\newacronym{TSO}{TSO}{Transmission System Operator}
\newacronym{BLO}{BLO}{BiLevel Optimization Problem}
\newacronym{GNE}{GNE}{Generalized Nash Equilibrium}
\newacronym{MPCC}{MPCC}{Mathematical Program with Complementarity Constraint}
\newacronym{GNEP}{GNEP}{Generalized Nash Equilibrium Problem}
\newacronym{vGNE}{\textit{v}-GNE}{\textit{variational} Generalized Nash Equilibrium}
\newacronym{lSE}{$\ell$-SE}{local Stackelberg equilibrium}
\newacronym{MPEC}{MPEC}{Mathematical Program with Equilibrium Constraints}
\newacronym{MIP}{MIP}{Mixed-Integer Program}
\newacronym{KKT}{KKT}{Karush–Kuhn–Tucker}

\newtheorem{theorem}{Theorem}
\newtheorem{lemma}{Lemma}
\newtheorem{definition}{Definition}
\newtheorem{assumption}{Assumption}

\title{\LARGE \bf
A Stackelberg game for incentive-based\\ demand response in energy markets}

\author{Marta Fochesato, Carlo Cenedese, John Lygeros% <-this % stops a space
\thanks{Authors are with Automatic Control Laboratory, Department of Electrical Engineering and Information Technology,
        ETH Z\"urich, Physikstrasse 3 8092 Z\"urich, Switzerland
        {\tt\small \{mfochesato, ccenedese, jlygeros\}@ethz.ch}.
}
\thanks{        Research is supported by BFE and the ETH Foundation under the ReMaP project and  by NCCR Automation, a National Centre of
Competence in Research, funded by the Swiss National Science
Foundation (grant number 180545).}%
}

\begin{document}

\maketitle
\thispagestyle{empty}
\pagestyle{empty}

% in the abstract
\begin{abstract} 
In modern buildings renewable energy generators and storage devices are spreading, and consequently  the role  of the users in the power grid is shifting from passive to active. We  design a demand response scheme that exploits the prosumers' flexibility to provide ancillary services to the main grid. We propose a  hierarchical scheme to coordinate the interactions between the distribution system operator and a community of smart prosumers. The framework inherits characteristics from price-based and incentive-based schemes and it retains the advantages of both. We cast the problem as a Stackelberg game with the prosumers as followers and the distribution system operator as leader. We solve the resulting bilevel optimization program via a KKT reformulation, proving the existence and the convergence to a local Stackelberg equilibrium. %Finally, we provide numerical simulations to corroborate our claims on the benefits of the proposed framework.

%\red{General: \begin{itemize}
%    \item Increase the font size of all the figures cause rn they are too small compared to the one in the text.
%    \item Careful w the titles in the literature. Several miss capitalization in the right words
%    \item Never use $\forall$ inline. Online inside an equation.
%\end{itemize}}
\end{abstract}

\IEEEpeerreviewmaketitle

\section{Introduction}
The increase in the share of renewable generation and electric storage devices, together with the deployment of smart buildings capable of  purchasing and selling energy to the main grid, has lead to an increasing active role of end users, turning them from consumers to \textit{prosumers}. This shift has been acknowledged also by the latest European directives~\cite{EU} that highlight the need for a holistic scheme that integrates local markets, wholesale markets and the provision of ancillary services. In this context, the concept of \gls{DR} is becoming popular as a way of harnessing prosumers flexibility to provide ancillary services to the main grid. \gls{DR} encompasses changes in the electric usage by end users induced, for example, by changes in the price of electricity over time (\textit{price-based} schemes) or monetary incentives (\textit{incentive-based} schemes) \cite{Overall}. Price-based schemes are extensively studied in literature (see \cite{Price1}, \cite{Price2} and the references therein), but incentive-based  ones are gaining traction thanks to the greater freedom  provided to the end users who may or may not accept the incentive \cite{Pricing_review}. 

As entities involved in the \gls{DR} scheme, we consider  the \gls{TSO}, the \gls{DSO} and a collection of prosumers. This gives rise  to a hierarchical structure that can be cast as a bilevel optimization problem (see \cite{Overall} for a complete survey)  resulting from a Stackelberg game between one (or multiple) leader (the \gls{DSO}) and multiple followers (the prosumers) \cite{Overall2}. Different specialized \gls{DR} schemes have been proposed in the literature, dealing with either the wholesale market for pricing purposes, or with the local one for congestion management. To the best of our knowledge, an integrated market design coupling the two markets as envisioned by the European directives in \cite{EU} is still missing. 

We propose a novel Stackelberg-based \gls{DR} model that unifies a time-of-use formulation for pricing and a emergency \gls{DR} scheme with the goal of  maintaining the stability of the grid by satisfying a given flexibility request, both in terms of (upward) response and (downward) rebound. This is achieved via a multi-objective formulation where the pricing scheme serves a peak shaving objective, while the incentive scheme serves the ancillary service provision objective. In particular, we envision the case where the monetary incentive is governed by a contract between \gls{DSO} and prosumers ensuring fairness by distributing the incentive proportionally to the effort sustained by the prosumers. 
%In particular, we envision the case where the monetary incentive is proportional to the effort sustained by the prosumers as this would lead to a fair market mechanism. 

Our contribution can be summarized as follows.
\begin{itemize}
    \item We design an integrated \gls{DR} scheme coupling the pricing map selection problem (wholesale market) with the flexibility provision problem (ancillary service). Unlike \cite{Comp1} and \cite{Comp2}, we design \textit{soft} reward functions capable of describing saturation to avoid infeasibility and provide a broader action space for the prosumers. 
    \item We cast the \gls{DR} scheme as Stackelberg game, which we turn into a \gls{MPEC} by embedding the equivalent KKT conditions of the followers into the leader optimization problem  to compute a \gls{vGNE} of the followers' game.
    \item We prove the existence of  {at least one} \gls{lSE} for the game and that the strategies of the leader and the followers converge to {it}.
\end{itemize}
\mbox{The notation adopted is borrowed from \cite{Filippo}.}

\section{Market design} \label{Sec:DR}
%For the sake of clarity of exposition, we first propose a high-level formulation of the whole \gls{DR} scheme, and then focus on modeling the individual components.

We consider a single \gls{DSO} purchasing flexibility services from a community $\mathcal{N} \coloneqq \{1, \ldots , N\}$ of prosumers. The structure of the problem is inherently  hierarchical \cite{Overall2}. We consider a day-ahead scheduling problem where trading takes place over $T$ intervals of equal length {$\Delta\tau > 0$}, i.e., the scheduling period runs over $\tau \in \{1,\dots, T\}$.  At the  beginning of the scheduling period, the \gls{TSO} sends to the \gls{DSO} a flexibility request signal $r\coloneqq \text{col}((r_\tau)_{\tau \in\{1,\dots ,T\}})\in\R^T$  for the upcoming $T$ time intervals. Additionally, the \gls{TSO} communicates the associated reward function \mbox{$\pi(\cdot|{r}): \mathbb{R}^T \rightarrow \mathbb{R}_{\geq 0}^T$} over the entire scheduling period. The latter defines the monetary incentive received by the \gls{DSO} for the flexibility provided, given the request $r$.

The \gls{DSO}'s goal is assumed to be to  maximize its revenue. In this direction, the role of the \gls{DSO} is twofold. First, it designs the function $h:\R^T\rightarrow \R^T_{\geq 0}$ that defines the price of the energy purchased by the prosumers during each time interval $\tau$. Additionally, the \gls{DSO} redistributes part of the revenues $\pi(\cdot|r)$ to those prosumers actively offering flexibility according to a predefined agreement.   

Given the pricing map and the redistribution share, the prosumers compute the optimal amount of purchased power and flexibility provided in order to minimize their individual economic cost.  Additionally, each follower $i\in\mc N$ has a set of local and coupling constraints; the former rules the energy management of the $i$-th agent, while the latter limits the aggregated power demand to be within the grid capacity. The choice  of each $i\in\mc N$ influences that of the others, thus the arising decision-making process  is a noncooperative constrained game. 

This hierarchical setup  can be efficiently modeled by means of a single-leader (the \gls{DSO}) multi-follower (the prosumers) Stackleberg game, as done also in \cite{Overall2} and \cite{BLOalgo}, resulting in a bilevel problem of the form,
\begin{equation}
\label{eq:stack_game}
\begin{split}
&\argmin_{z_0\in\mc Z,\, \boldsymbol{x} }
J^\tup{{DSO}}(z_0,\bs x| r)\\
&\: {\text{s.t.}} \quad x_i \coloneqq \argmin_{x_i\in\mc X_i(\boldsymbol{x}_{-i})} J^i(x_i,\bs{x}_{-i}|z_0),\: \forall i\in \mc N,
\end{split}
\end{equation}
where $z_0 \in \mathbb{R}^{n_0}$ groups the decision variables controlled by the \gls{DSO} and $x_i \in \mathbb{R}^{n_i}$ those of prosumer $i\in\mc N$, $J^\tup{{DSO}}$ and $J^i$ are the cost of the \gls{DSO} and of each prosumer $i$, and the sets $\mc Z$ and $\mc X_i$ are the feasible sets for the decision variables $z_0$ and $x_i$, respectively. % A representation of the interactions between the various entities in the game is reported in Figure \ref{fig:draw}.

%\begin{figure}
%    \centering
%    \includegraphics[trim={10 10 10 0},clip,width=0.9\linewidth]{draw_new.png}
%    \caption{Interactions between \gls{TSO}, \gls{DSO} and prosumers.}
%    \label{fig:draw}
%\end{figure}

\subsection{Reward function $\pi$ and revenue redistribution}
{The scheduling period's intervals} are classified into three mutually exclusive classes (see \cite{Rebound}): intervals with \mbox{$r_\tau = 0$} (no ancillary service required by the \gls{TSO}), \mbox{$r_\tau > 0$} ({\em response blocks}, where the \gls{TSO} requests reduction in energy consumption), and those with \mbox{$r_\tau < 0$} ({\em rebound blocks}, where the \gls{TSO} requests an increase in energy consumption). We assume that any load reduction or increase causes a deviation from the optimal operation  of the connected devices, leading to additional costs. To compensate, if the \gls{DSO} manages to provide the flexibility to the \gls{TSO}, it receives an economic reward that it partly share{s} with the prosumers.

%\begin{figure}[b]
%    \centering
%    \includegraphics[trim={50 0 80 20},clip,width=\linewidth]{r_END.pdf}
%    \caption{Example of a \gls{DR} request signal ${r}$ received from the \gls{TSO}. The red areas denotes rebound blocks $r<0$, and the green ones response ones $r>0$; areas where no ancillary services are needed are in white.}
%   \label{fig:r}
%\end{figure}

%For ease of exposition, l
Let us divide the  cases of response and rebound by means of two variables \mbox{$y_i \in  \mathbb{R}_{\geq 0}^T$} and \mbox{$k_i \in \mathbb{R}_{\leq 0}^T$}, denoting the amount of flexibility provided in the two cases respectively by follower $i\in\mc N$. 
During response blocks, the \gls{DSO} receives a monetary incentive \mbox{$\pi^\tup{R}(\cdot|r) : \mathbb{R}^T_{\geq 0} \rightarrow \mathbb{R}^T$} dependent on the aggregate response flexibility $\sum_{i=1}^N{y_i}$ provided. $\pi^\tup{R}(\sum_{i=1}^N{y_i} | {r})$  is defined component wise for each \mbox{$\tau\in\{1,\ldots,T\}$  as}
\begin{equation}\label{eq:reward_DSO_R}
   \begin{cases}
   0, \hspace{1.2cm}\:\quad & \text{if}\: r_\tau \leq 0\\
    \overline{p} \sum_{i=1}^N{y_{i,\tau}}, \quad  &\text{if}\: \sum_{i=1}^N{y_{i,\tau}} \leq {r_\tau}\\
    (\overline{p} -N\beta)\sum_{i=1}^N{y_{i,\tau}} + N\beta r_\tau, \quad  &\text{otherwise}
    \end{cases}
\end{equation}
where $\overline{p}>0$ is the price paid by the \gls{TSO} to the \gls{DSO} for each unit of response provided. The saturation coefficient  \mbox{$\beta\geq \overline p/N$}  decreases the monetary incentive rate after reaching the required amount, effectively discouraging the \gls{DSO} from offering more flexibility to the \gls{TSO}. Notice that when $\beta= \overline p/N$ we retrieve a simple saturation.

To persuade the followers to change their consumption pattern and provide the  flexibility $\sum_{i=1}^N{y_i}$, the \gls{DSO} shares with them part of the total reward $\pi^\tup{R}$. We assume a dynamic revenue share composed of two parts: a coefficient $\alpha\in[0,1]^T$ that modulates the incentive distributed to the followers, and a function 
\mbox{$\phi^\tup{R}_i\left(\cdot |{r}\right):\R^{NT}\rightarrow \R^T$} modeling the agreement between \gls{DSO} and prosumers on how such an incentive is designed.
% We assume a dynamic revenue share via a coefficient $\alpha\in[0,1]^T$ that modulates the incentive distributed to the followers.
Therefore, the incentive received by each prosumer $i\in\mc N$ from the \gls{DSO} is given by $\Lambda \phi^\tup{R}_i\left(\left.y_i, \bs{y}_{-i}\right|{r}\right)$, where \mbox{$ \Lambda\coloneqq\text{diag}(\alpha)\in\R^{T\times T}$} and \mbox{$\phi^\tup{R}_i$} is defined component-wise for each $\tau \in \{1,\ldots, T\}$ as
\begin{equation}\label{eq:rew_i}
   % \phi^\tup{R}_i\left(\left.y_i, \boldsymbol{y}_{-i}\right|{r}\right)\coloneqq
   \begin{dcases}
    0, \hspace{1.2cm}\qquad \text{if}\: r_\tau \leq 0 \\
    \overline{p} y_{i,\tau}, \qquad\qquad \text{if} \: y_{i,\tau} \leq {r_\tau}- \textstyle \sum_{j \in\mc N\setminus\{i\}} y_{j,\tau}&\\
    (\overline{p} - \beta) y_{i,\tau} + \beta \left({r_\tau} - \textstyle \sum_{j \in\mc N\setminus\{i\}}y_{j,\tau}\right), \: \text{oth.}& 
    \end{dcases}
\end{equation}
Notice that this choice ensures a fair share of the revenue redistribution among prosumers {proportional to} the flexibility provided, and  that satisfies  $ \pi^\tup{R }\left( \left.\sum_{i=1}^N{y_i}\right| {r}\right)= \sum_{i=1}^N \phi^\tup{R}_i\left(\left.y_i, \bs{y}_{-i}\right|{r}\right)$.
During a response block, the net reward collected by the \gls{DSO} for the \gls{DR} provision is then a piece-wise linear function of $\sum_{i=1}^N{y_i}$, i.e.,
\begin{equation}\label{eq:net_DSO_R}
   \pi^\tup{{net,R}}\left( \left.\sum_{i=1}^N{y_i}\right| {r}\right)= (I_{T} - \Lambda) \sum_{i=1}^N \phi^\tup{R}_i\left(\left.y_i, \bs{y}_{-i}\right|{r}\right).
\end{equation}
%where we used  that follows directly form \eqref{eq:reward_DSO_R} and \eqref{eq:rew_i}.
%The reader may easily check that an overall reward balance for response blocks holds such that no economic incentive is dissipated nor created outside of $ \pi^\tup{R}\left( \left.\sum_{i=1}^N{y_i}\right| {r}\right)$. 
A qualitative  representation of $\pi^\tup{{net,R}}$ and $\phi^\tup{{R}}$ is depicted in Figure \ref{fig:functions} for different values of $\alpha$ and $\beta$. 
The formulation \eqref{eq:rew_i}  reflects the saturation condition in~\eqref{eq:reward_DSO_R}.
%The formulations in \eqref{eq:reward_DSO_R} and \eqref{eq:rew_i}  model a saturation condition intended to  discourage both the \gls{DSO} and the followers from offering more flexibility than the required one, without explicitly preventing them in doing so. 
This allows a more responsive energy management for the end users and it prevents infeasibility in the optimization problem.
\begin{figure}
    \centering
    \includegraphics[trim={50 40 50 20},clip,width=\linewidth]{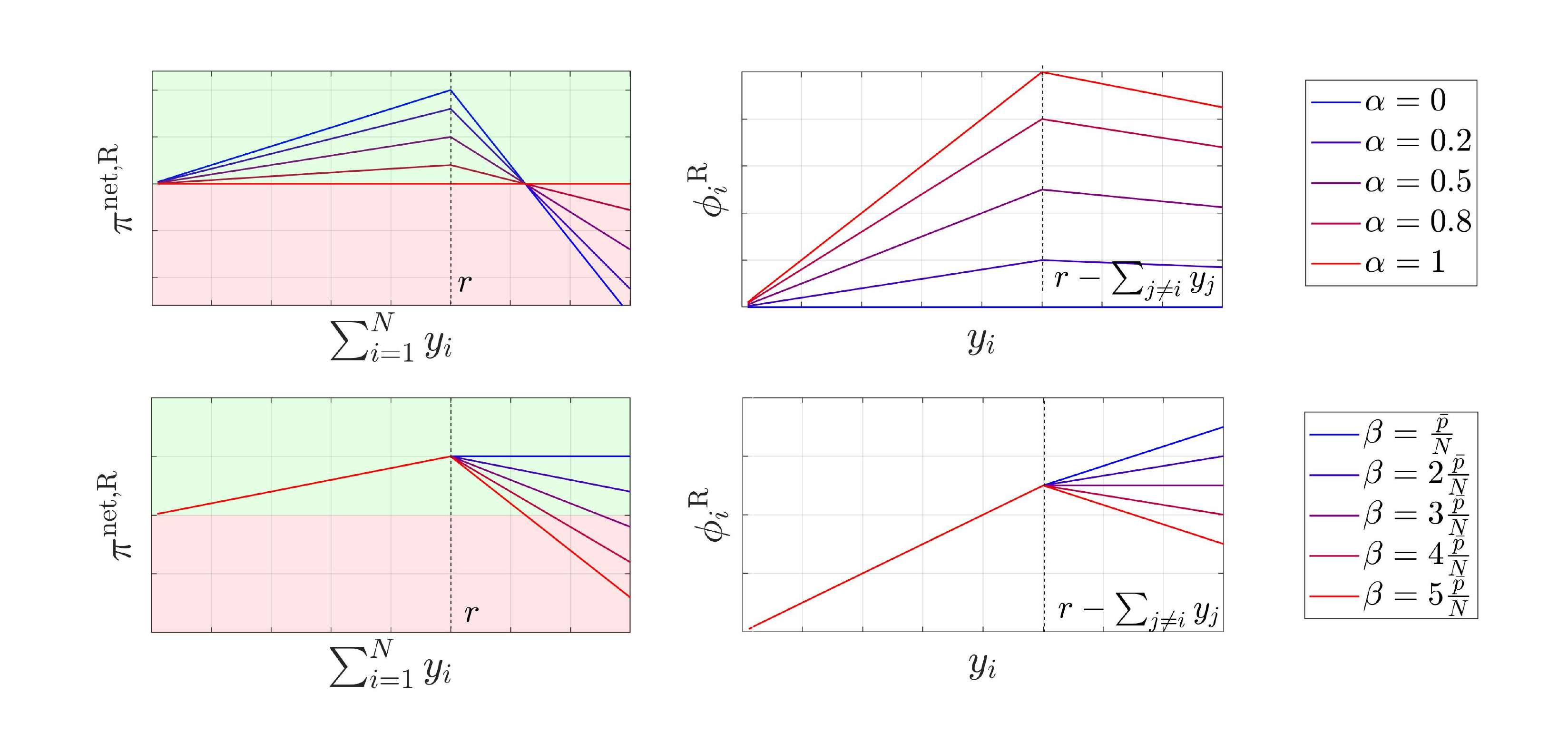}
    \caption{Qualitative behaviour of the reward functions for the \gls{DSO} (left) and the followers (right) for different $\alpha$ (top) and $\beta$ (bottom). Red areas correspond to a cost, while green ones correspond to a reward.}
    \label{fig:functions}
\end{figure}

Similarly, during rebound blocks, the \gls{DSO} receives an economic incentive $\pi^\tup{B}(\cdot| {r}) : \mathbb{R}^T_{\leq 0} \rightarrow \mathbb{R}^T$ proportional to the aggregate rebound flexibility $\sum_{i=1}^N{k_i}$ provided:
\begin{equation}\label{eq:reward_DSO_B}
   \pi^{\tup{B}}\left(\left.\sum_{i=1}^N{k_i}\right| {r}\right)\coloneqq \begin{dcases} 0, \quad \hspace{1.6cm}\: \text{if} \: r \geq 0\\
   -\tilde{p} \textstyle \sum_{i=1}^N{k_i}, \quad \: \text{otherwise},
   \end{dcases}
\end{equation}
where $\tilde{p}$ is the price provided by the \gls{TSO} to the \gls{DSO} for unit of rebound.
From the prosumers point of view, during rebound blocks the prosumers are asked to absorb  extra  power  from the main grid to maintain its stability. In this case, we assume that the prosumers can access the $|k_i|$ unit of power from the main grid at zero cost. Since they will save to purchase this additional power from the grid paying $h$, the \gls{DSO} does not need to share part of the revenues with the prosumers to make them contribute to the rebound. 
%the reward of the followers is given by the possibility of accessing $k_i$ amount of power at zero cost from the main grid which can be stored in the storage device for later use. 
Consequently, the reward is entirely kept by the \gls{DSO} for its intermediary service. The current mathematical description of the rebound {does not guarantee  an increment} in the amount of power taken from the grid with respect to a baseline case. Nevertheless, we notice from simulations that prosumers do take advantage of rebound blocks to absorb more energy. This is a behaviour akin to the one encountered in peak-shaving schemes \cite{Price2}.

\subsection{\gls{DSO} constraints and cost function}\label{subsec:DSO}
The \gls{DSO} aims solely at maximizing its revenues \cite{Overall}. These derive both from selling energy and from providing flexibility to the \gls{TSO}. The pricing map $h$ is assumed to be affine with respect to the total energy purchased by the prosumers, as it is shown in \cite{Price2},\cite{cenedese:2019:PEV_charging} that this map induces a peak-shaving behaviour.
\smallskip
\begin{assumption}[Affine pricing map]\label{ass:price}
For an aggregate purchased power $\sum_{i=1}^N p_i$, the pricing map $h:\R^T_{\geq 0}\rightarrow \R^T$ is
\begin{equation}
    h\left(\sum_{i=1}^{N} p_i\right) = C_1\sum_{i=1}^{N}p_i + c_0,
\end{equation}
where $C_1 = \text{diag}(c_1)$ with $c_1, c_0 \in \mathbb{R}^T_{\geq 0}$.\hfill\QEDopen
\end{assumption}
\smallskip
While $c_0$ is a parameter that is optimized by the \gls{DSO}, $c_1$ is assumed to be fixed and used to model the peak price throughout the prediction horizon. This assumption prevents the presence of tri-linear terms in the cost function that would highly complicate the  analysis.

%We do not impose any constraints on how fast the price coefficients may vary. However, one can add local constraints for upper bounding such rate with minor modifications and without affecting the theoretical framework proposed.

A second stream of revenues for the \gls{DSO} comes from the reward for the provision of \gls{DR} services that is not redistributed to the agents, viz. $\mathds{1}_T^\top \pi^{\text{net,R}}\left(\left.\sum_{i=1}^N{y_i} \right| r\right)$  and $\mathds{1}_T^\top \pi^\tup{B}\left(\left.\sum_{i=1}^N{k_i} \right| r\right)$ . Therefore, the resulting cost function for the \gls{DSO} reads as
\begin{equation}\label{eq:cost_DSO_ori}
\begin{aligned}
    {J}^{\tup{DSO}} \coloneqq  -\left(C_1 \sum_{i=1}^{N}p_i + c_0\right)^\top \sum_{i=1}^{N}p_i - \mathds{1}_T^\top (\pi^{\tup{net,R}} + \pi^\tup{B}),
    \end{aligned}
\end{equation}
where we dropped the arguments of the rewards for the sake of a lighter notation.

Note that in \eqref{eq:cost_DSO_ori} we omit the cost that the \gls{DSO} has to incur to purchase power from the main grid. Alternatively, if one considers the electricity selling price $s$ from the main grid to the \gls{DSO} to be proportional (or affine) to the aggregate power $\sum_{i=1}^N p_i$, then it is possible to incorporate it into the first term of (\ref{eq:cost_DSO_ori}), without changing the problem \mbox{formulation.}

%Without loss of generality, we can divide $\hat{J}^{\tup{DSO}}$  by $c_1$. In this  reformulation, the new decision variables \mbox{$v_1, v_2 \in \mathbb{R}^T$} are defined as \mbox{$v_{1,\tau} \coloneqq \frac{1}{c_{1,\tau}}$} and \mbox{$v_{2,\tau} \coloneqq \frac{\alpha_\tau}{c_{1,\tau}}$} for all \mbox{$\tau \in \{1,\ldots, T\}$.} The cost becomes
%\begin{equation}\label{eq:cost_DSO}
%\begin{aligned}
%    J^{\tup{DSO}} \coloneqq &  -\left( \sum_{i=1}^{N}p_i + V_1 c_0\right)^\top \sum_{i=1}^{N}p_i\\ & - \mathds{1}_T^\top \left( V_1 \pi^\tup{R} \B{-} V_2 \sum_{i=1} ^N \phi^\tup{R}_i 
%    \B{+} V_1\pi^\tup{B}\right),
%\end{aligned}
%\end{equation} 
%where $V_1\coloneqq \text{diag}(v_1)$ and $V_2\coloneqq\text{diag}(v_2)$.
The variables $c_0, \alpha \in \mathbb{R}^T$ are subject to the following box constraints, $\underline{c_0} \leq c_0 \leq \overline{c_0}$ and $ 0 \leq \alpha \leq 1$. Here, the constraints on $c_0$ are used to prevent  the \gls{DSO} from raising the price coefficients unrealistically, since we are modeling a monopoly. The constraints over $\alpha$ ensure that the term is indeed a fraction.
%\begin{subequations}
%\label{eq:cnst_v_1_v_2}
%\begin{align}
%   & \underline{c_0} \leq c_0 \leq \overline{c_0}\\
%   &    0 \leq \alpha \leq 1.
%\end{align}
%\end{subequations}

\subsection{Prosumer constraints and cost function}

We assume that each prosumer aims to meet its power demand at the minimum cost. Therefore, the cost function of each prosumer is the sum of four contributions 
\small{
\begin{equation}\label{eq:cost_i}
\begin{aligned}
    J^{i} \coloneqq & \underbrace{\left(C_1 \sum_{i=1}^{N}p_i +  c_0\right)^\top p_i}_{\text{power purchasing cost}} - \underbrace{\vphantom{  \left(C_1 \sum_{i=1}^{N}p_i +  c_0\right)^\top} \mathds{1}_T^\top \bigg((I_T - \Delta) \phi_i  ^{\tup R}\bigg)}_{\text{response reward}}
     \\
     & + \underbrace{\vphantom{\left(C_1 \sum_{i=1}^{N}p_i +  c_0\right)^\top} \delta \left( p_i^\tup{C} + p_i^\tup{DC}\right)}_{\text{storage degradation cost}} + \underbrace{\vphantom{\left(C_1 \sum_{i=1}^{N}p_i +  c_0\right)^\top} \mu y_i}_{\text{discomfort cost}},
    \end{aligned}
\end{equation}}
\normalsize
where $\delta>0$ is a degradation weight and $0 \leq \mu << \bar{p}$ is a (small) discomfort weight. For each time interval $\tau \in  \{1,\ldots, T\}$, the $i$-th prosumer must schedule its electric storage and define the flexibility to provide when requested. The affine local set of constraints describing the dynamic evolution of the state of charge of the storage device reads as 
\begin{subequations}\label{eq:cons_i}
\begin{align}
    & e_{i,\tau} = e_{i,\tau-1} + \Delta \tau\big(\eta^\tup{C} p^\tup{C}_{i,\tau} -  \eta^\tup{DC}p^\tup{DC}_{i,\tau}\big),\\
  & 0 \leq p^\tup{C}_{i,\tau} \leq p^{\max}_i,  \\
  & 0 \leq p^\tup{DC}_{i,\tau} \leq p^{\max}_i, \\
  & 0 \leq e_{i,\tau} \leq e^{\max}_i,\\
 & p_{i,\tau} - k_{i,\tau} \geq d_{i,\tau} - s_{i,\tau} + p^\tup{C}_{i,\tau} - p^\tup{DC}_{i,\tau},
\end{align}
\end{subequations}
where $e_i$ is the state of charge of the electric storage device and $e_{i,0}$ its initial state, while $\eta^C$ and $\eta^{DC}$ are the charging and discharging efficiency of the storage device, respectively.
The prosumer's {fixed} load is denoted by $d_i$, while $s_i$ stands for the produced renewable power. We assume that the local demand of each prosumer is higher or equal
than the local renewable power generated. Thus, we do not model sell-back programs and/or curtailment. As a consequence, $d_i - s_i + p_i^\tup{C} - p_i^\tup{DC} \geq 0$, thus $p_i \geq 0$.

\subsection{Coupling constraint}
Finally, the last set of constraint couples the energy consumption of prosumers:
\begin{subequations} \label{eq:coupling}
\begin{align}\label{eq:13a}
 & \sum_{i=1}^{N}p_i + \sum_{i=1}^{N}y_i  - \sum_{i=1}^{N}k_i\leq \max (g, g-r)\\ % \\ \label{eq:13b}
 %& \sum_{i=1}^{N}p_i - \sum_{i=1}^{N}k_i \leq g - r\\
 \label{eq:13c}
     & \sum_{i=1}^N |k_i| \leq |r|,
    \end{align}
\end{subequations}
where $g \in \mathbb{R}_{\geq 0}^T$ is the resource vector representing the grid capacity. Loosely speaking, in \eqref{eq:13a} the grid capacity $g$ is artificially reduced  during the \gls{DR} response blocks of a quantity $\sum_{i=1}^{N}y_i$. On the other hand, during rebound blocks $g$ is physically enlarged. We stress the difference between response which results in a \textit{virtual} restriction of the resource vector and rebound which implies a \textit{physical} enlargement of it. Finally, \eqref{eq:13c} prevents %is needed to restrict 
prosumers from taking more units of energy for free than the ones requested by the \gls{TSO} rebound signal. 

\subsection{Discussion of alternative design choices}
Throughout Section I, we have introduced several design choices. We briefly motivate them here, comparing them with possible alternatives. Alternative market designs are:
\begin{itemize}
    \item The \gls{DSO} sells power at a discounted price during rebound blocks instead of for free. This would increase the symmetry between the response and rebound blocks.
    \item Some users are still allowed to absorb extra amount of power during response blocks, as long as on a global scale there is still a reduction (and viceversa for rebound blocks). This would increase the action space of the prosumers, leading to possibly lower costs.
    \item The \gls{DSO} dynamically adjusts the coefficients in the $h$ function to encourage prosumers to offer flexibility without the need of a separate incentive scheme. 
\end{itemize}
While the first two choices provide similar formulations to the proposed one in terms of theoretical properties and market operations, the third choice would lead to a different market structure where the benefits of providing flexibility are global (i.e., via the pricing map $h$) instead of proportional to the individual effort sustained by each prosumer as proposed. This might lead to unfair situations with the presence of free-riders. %Designing a pricing map $h$ capable of simultaneously accounting for peak shaving and for ``transparent" benefits is a challenging task which would compromise the theoretical properties necessary to establish existence and convergence to a \gls{lSE}.

\section{Stackelberg game} \label{Sec:SG}
\subsection{Reformulation as bilevel optimization program}
We are now ready to recast the  \gls{DR} scheme as a Stackelberg game with one leader and $N$ followers as in \eqref{eq:stack_game}. 

The affine reward functions for the response, $\pi^\tup{R}(\sum_{i=1}^N y_i, | r)$ and $\phi^\tup{R}_i( y_i, \boldsymbol{y}_{-i} | r)$, introduce nonlinearities. To proceed with the analysis, we implement an epigraph reformulation and replace them with auxiliary variables.  As an illustrative example, let us consider the term $\min_{y_i}-\phi^\tup{R}_i(y_i, \boldsymbol{y}_{-i} | r)$ in \eqref{eq:cost_i} for a fixed time step $\tau$. It can be shown \mbox{\cite[Sec.~4.1]{Boyd}} that it is equivalent to the following epigraph reformulation 
\begin{equation}\label{eq:epi}
\begin{array}{l}
\underset{y_{i,\tau}, t_{i,\tau}}{\min} \hspace{0.2cm} t_{i,\tau} \\
\text{s.t.} 
\hspace{0.4cm} t_{i,\tau} \leq 0\\
\hspace{0.8cm} t_{i,\tau} \geq -\bar{p} y_{i,\tau}\\
 \hspace{0.8cm} t_{i,\tau} \geq -\left(  (\bar{p} - \beta) y_{i,\tau} + \beta \left({r} - \textstyle \sum_{j \in\mc N\setminus\{i\}}y_{j,\tau}\right) \right),
\end{array}
\end{equation}
where $t_{i,\tau} \in \mathbb{R}$ is an auxiliary variable. From the point of view of the leader, this reformulation simply amounts to replacing \eqref{eq:net_DSO_R} by $-(I_{N} - \Lambda) \textstyle \sum_{i=1}^N t_i$. The rebound case, on the other hand, does not present any nonlinearity in the domain where $k_i$ is defined, see~\eqref{eq:reward_DSO_B}. % Therefore, an epigraph reformulation is not needed. 
Note that the epigraph reformulation above introduces a series of  additional local and coupling constraint inequalities. 

%Let us now
Next, we define the feasible decision  sets for both the leader and the followers. Starting from the \gls{DSO}, the collective leader strategy $z_0 \coloneqq \text{col}(c_0, \alpha)$ has to satisfy box constraints (cfr. Subsection \ref{subsec:DSO}); thus we can write the \gls{DSO} feasible set  compactly as
 \begin{equation}\label{eq:cons_DSO}
\begin{aligned}
 {\Gamma}\coloneqq \left\{ z_0 \in \mathbb{R}_{\geq 0}^{2T} \: : \: F^\tup{DSO}z_0\leq g^\tup{DSO}\right\}.
\end{aligned}
\end{equation}
Note that $\Gamma$ is nonempty if $\overline{c_0} \geq \underline{c_0}\geq 0$ and it is  compact. 
The local feasible decision of the \mbox{$i$-th} prosumer is denoted by \mbox{${x}_i \coloneqq \text{col}(p_i, y_i, e_i, p_i^{\tup{C}}, p_i^{\tup{DC}}, k_i, t_i) \in\R^{7T}$}. Prosumer $i$ has to satisfy 
\eqref{eq:cons_i}, \eqref{eq:epi} and the non-negativity (respectively, non-positivity) bounds on $y_i$ (respectively, $k_i$) {represented in compact form as} as \mbox{$F_i x_i \leq f_i$}. Therefore the local feasible decision set is \mbox{$\Omega_i\coloneqq \left\{ x_i\in \R^{7T} : F_i x_i \leq f_i \right\}$}, where $\Omega_i$ is a compact set as all variables are bounded (see \eqref{eq:cons_i}, \eqref{eq:coupling} and recall the no sell-back programs assumption). Additionally, the coupling constraints in  \eqref{eq:coupling} can be compactly written as ${A}_i {x}_i + \sum_{j \in \mathcal{N}} {A}_j {x}_j  \leq {b}$. Therefore, the followers feasible decision set becomes
\begin{equation}\label{mathX}
{\mathcal{X}}_i(\bs{x}_{-i})\coloneqq \left\{ x_i\in\Omega_i : {A}_i {x}_i + \sum_{j \in\mc N\setminus\{i\}} {A}_j {x}_j  \leq {b}\right\}.
\end{equation}

We introduce some blanket assumptions on this set of feasible strategy,
standard in the literature \cite{Filippo}.
\smallskip
\begin{assumption}\label{Slater}
For each player $i \in \mc N$, the set $\Omega_i$ is nonempty. The collective feasible set $\bs{\mc X}\coloneqq \prod_{i \in \mathcal{N}}\mathcal{X}_i(\bs{x}_{-i})$ satisfies Slater's constraint qualification.
\end{assumption}
\smallskip
Note that, Slater's constraint qualification can be easily met by choosing a resource vector $g$ large enough.

%For technical reasons that will be clarified in the following, we introduce a quadratic regularization term $x_i^\top \text{diag}(\gamma) x_i$ in \eqref{eq:cost_i} where $\gamma \in \mathbb{R}_{> 0}^{7T}$  \B{ensures the} strong convexity of the \B{prosumers}'s cost function as \B{proven} in Lemma 1. 

We reorganize the cost functions \eqref{eq:cost_DSO_ori} and \eqref{eq:cost_i} in a compact matrix form as
\begin{equation}\label{eq:cost_DSO_matrix}
    J^{\tup{DSO}}({z}_0, \boldsymbol{{x}}) = {f_0({z}_0)} + {f_x(\boldsymbol{{x}})} + {\left( \sum_{i \in \mathcal{N}}f_{0,i}({x}_i)\right)^\top {z}_0},
\end{equation}
% \small{\begin{equation}
% \begin{aligned}\label{eq:cost_i_matrix}
%     J^i\left( {x}_{i}, \boldsymbol{{x}}_{-i},{z}_{0}\right) = & \frac{1}{2}{x}_i^T {Q}_i {x}_i +  \left(\sum_{j \in \mathcal{N}} C_{i,j} {x}_j + C_{i,0} {z}_0\right)^\top {x}_i ,\\
% \end{aligned}
% \end{equation}}
{\small \begin{equation}
\begin{aligned}\label{eq:cost_i_matrix}
    J^i\left( {x}_{i}, \boldsymbol{{x}}_{-i}|{z}_{0}\right) = & {x}_i^T Q {x}_i +  \left(\sum_{j \in\mc N\setminus\{i\}} Q {x}_j + C_{i,0} {z}_0\right)^\top {x}_i ,\\
\end{aligned}
\end{equation}}
for an appropriate choice of $f_0$, $f_x$, $f_{0,i}$, and $ C_{i,0}$. For all $j,k\in\{1,\cdots,7T\}$, $Q$ is defined component wise as \begin{equation}\label{eq:Q_def}
[Q]_{jk} \coloneqq
\begin{cases}
1 & \text{if } j=k \text{ and } j\leq T\\
0 & \text{oth.}
\end{cases}.
\end{equation} 

At this point, we have recast the Stackelberg game between the \gls{DSO} and the $N$ prosumers as a bilevel optimization problem as anticipated in \eqref{eq:stack_game}. 

\subsection{ Stackelberg equilibrium}
Next, we discuss the equilibrium concept associated to the Stackelberg game in~\eqref{eq:stack_game}. 
For a fixed strategy of the leader, $z_0$, the followers take part in a  \gls{GNEP}. In general, solving a  \gls{GNEP} is a difficult task and does not allow for a nice formulation of the set of all solutions that can be exploited by the leader during the optimization. For this reason, we focus on the subset of \gls{vGNE}.  This set of equilibria represent a ``fair'' competition among followers, since they equally share the cost of satisfying the coupling constraints,  see \cite{VI}. The set of \gls{vGNE}  can be characterized as the solution set of variational inequalities $\mathrm{VI}(\bs{\mc X}, H(z_0, \cdot)) $, where \mbox{$H(z_0, \boldsymbol{x}) \coloneqq \text{col}((\nabla_{x_i} J^i(z_0, \boldsymbol{x}))_{i \in \mathcal{N}})$} is the \textit{pseudo-gradient} mapping of the followers game \cite{VI}. 
 For any given $z_0$, this set of \gls{vGNE} reads as,
\begin{equation}\label{eq:S}
  \mathcal{S}(z_0) \coloneqq \{ \boldsymbol{x} \in \bs{\mc X} \B{:} (\boldsymbol{w} - \boldsymbol{x})^\top H(z_0, \boldsymbol{x})  \geq 0, \forall \boldsymbol{w} \in \bs{\mc X}\}.
\end{equation}
Next, we show that this solution set is nonempty and compact. 

\smallskip

{\begin{lemma}[Existence of \gls{vGNE}]\label{lemma2}
The set of \gls{vGNE} $\mc{S}(z_0)$ in \eqref{eq:S}  is nonempty and compact for all $z_0\in\Gamma$.
\end{lemma}}

\smallskip
{\begin{proof}
{From \cite[Th.~6]{Exist} it follows that t}he continuity of $H(z_0,\cdot)$ and the compactness of  $\boldsymbol{\mathcal{X}}$ entails the existence of a solution to VI($H(z_0,\cdot),\boldsymbol{\mathcal{X}}$){, i.e., $\mc{S}(z_0)\neq \emptyset$}. Additionally \cite[Th.~2.2.5]{facchinei:pang:2003:finite_dim_VI} ensures the compactness of $\mc{S}(z_0)$. 
\end{proof} }
\smallskip

We can rewrite the {complete \gls{DR}} problem as
\begin{equation} \label{eq:Stack_S}
\left\{\begin{aligned}
\min _{z_{0}, \boldsymbol{x}} & \quad J^\tup{DSO}\left(z_{0}, \boldsymbol{x}\right) \\
\text { s.t. } & \quad \left(z_{0}, \boldsymbol{x}\right) \in {\gph(\mathcal{S}) \cap  (\Gamma\times \R^n)},
\end{aligned}\right.
\end{equation}
that is a \gls{MPEC}. {This class of} problems is inherently non-convex and typically hard to solve. In general, there is no guarantee that feasible solutions strictly lie in the interior of the feasible set, which may even be disconnected. As a consequence, constraint qualifications may be violated at every feasible point. For this reason we focus on seeking local solutions of \eqref{eq:Stack_S}. 

We consider the following definition of {local generalized Stackelberg equilibrium} from \cite{Filippo}.

\smallskip
\begin{definition}[Local generalized Stackelberg equilibrium]
A pair $(z_0^\star, \boldsymbol{x}^\star) \in \gph(\mathcal{S}) \cap(\Gamma\times\R^n)$, with $\mc S$ as
in \eqref{eq:S}, is a \gls{lSE} of the 
game in \eqref{eq:Stack_S} (and thus of \eqref{eq:stack_game}) if there exist open neighborhoods $\mathcal{O}_{z_0^\star}$ and $\mathcal{O}_{\boldsymbol{x}^\star}$ of $z_0^\star$ and $\boldsymbol{x}^\star$ respectively, such that
\begin{equation}\label{Def}
    J^\tup{DSO}\left(z_{0}^{\star}, \boldsymbol{x}^{\star}\right) \leq \inf _{(z_0,\bs x)\in \gph(\mathcal{S}) \cap (\Gamma\times\R^n) \cap \mathcal{O}} J^\tup{DSO}\left(z_{0}, \boldsymbol{x}\right)
\end{equation}
where $\mathcal{O}\coloneqq \left( \mathcal{O}_{z_{0}^{\star}} \times \mathcal{O}_{\boldsymbol{x}^{\star}} \right)$.
\end{definition}
\smallskip
Roughly speaking, at an \gls{lSE}, the \gls{DSO} and the prosumers locally fulfill the set of mutually coupling constraints and none of them can achieve a lower cost by unilaterally deviating from their current strategy.
\smallskip
\begin{theorem}[Existence of \gls{lSE}] 
\label{theorem} 
Under Assumption \ref{Slater}, there exists {at least one} \gls{lSE} of the \gls{MPEC} in \eqref{eq:Stack_S}. 
\end{theorem}
\begin{proof}  
 From Lemma~\ref{lemma2}, for all $\bar z_0\in\Gamma$ there exists an $\bar {\bs x}$ such that $(\bar z_0,\bar{\bs x})\in\gph(\mc S)$, therefore $\gph(\mathcal{S}) \cap (\Gamma\times\R^n)\neq \emptyset$ and it is compact since it is an intersection of compact sets.
 
% their intersection is a nonempty compact set \cite{Rock}. 
From \mbox{\cite[Th.~1.3.4]{Theo}} and the continuity of the cost function $J^{\tup{DSO}}$ with respect to $z_0$, we conclude the existence of a solution to the \gls{MPEC} in \eqref{eq:Stack_S}, that in turns satisfies the definition in \eqref{Def}. %

\end{proof}
\smallskip

\subsection{Solution method} \label{Sec:SM}
To solve the VI in \eqref{eq:Stack_S}, we exploit the strong relation between \glspl{vGNE} and KKT of the followers' game. The KKT conditions of the follower level read in a compact form as

\begin{equation}\label{eq:KKT}
    \left\{\begin{array}{l}\bs Q\boldsymbol{x} + Cz_0 +A^{\top} \lambda+F^{\top} \bs\lambda=0  \\
    0 \leq \lambda \perp-(A \boldsymbol{x}-b) \geq 0 \\ 
    {0 \leq \bs \lambda \perp-\left(F \bs x-f\right) \geq 0,}\end{array}\right.
\end{equation}
where $\lambda$ is the dual variable associated with the coupling constraints $A\boldsymbol{x} \leq b$. The (local) dual variable $\lambda_i$ is associated with the (local) constraints of each follower \mbox{$i\in\mc N$} and  $\boldsymbol{\lambda}\coloneqq \text{col}((\lambda_i)_{i\in\mathcal{N}})$, while $F\coloneqq \text{diag}((F_i)_{i\in\mc N})$, and \mbox{$f \coloneqq \text{col}((f_i)_{i\in\mc N})$}.  
 The matrices in the costs are
 \mbox{$\bs Q \coloneqq (I_N+\mathbb{1}_{N\times N})\otimes Q$} and  \mbox{$ C\coloneqq\left[\begin{array}{ccc}C_{1,0}^\top & \ldots & C_{N, 0}^\top \end{array}\right]^\top$.} 
{The following lemma ensures that if a collective strategy satisfies the {KKT} conditions, then it is also a \gls{vGNE}.}
\smallskip
\begin{lemma}[KKT reformulation \mbox{\cite[Th.~1.3.5]{Theo}}] \label{lemma3}
 Given any $z_0 \in \Gamma$, the strategy $\bs{x^\star}\in\mc S(z_0)$ if and only if there exists a triplet $(\bs{x^\star},\bs \lambda^\star,\lambda^\star)$ satisfying the {KKT} conditions in \eqref{eq:KKT}. 
\end{lemma}
\smallskip

Substituting \eqref{eq:KKT} back into \eqref{eq:Stack_S}, we obtain a single-level optimization problem 
\begin{equation}\label{eq:reference}
\begin{cases}\min _{y_{0}, \boldsymbol{x}, \lambda, \boldsymbol{\lambda}} & J^\tup{DSO}\left(z_{0}, \boldsymbol{x}\right) \\
\text { s.t. } & \bs Q \boldsymbol{x}+C z_{0} + A^{\top} \lambda+F^{\top} \boldsymbol{\lambda}=0 \\
& 0 \leq {\bs\lambda} \perp-\left({ F \bs x -f}\right) \geq 0, \\
& 0 \leq \lambda \perp-(A \boldsymbol{x}-b) \geq 0\\
& z_{0} \in \Gamma,\end{cases}
\end{equation}
that can be solved with one of the methods for MPEC proposed in \cite{luo_pang_ralph_1996}.

\section{Simulations} \label{Sec:Sim}
We apply our hierarchical \gls{DR} scheme to a community of residential buildings equipped with an array of photovoltaic panels and a battery for electric storage.  We consider a small-scale community of $N=5$ residential buildings during the heating season. Building loads are simulated in EnergyPlus via the Cesar-p interface \cite{cesar} using building characteristic input files corresponding to different prototypical Swiss residential buildings. Weather data are from MeteoSwiss, while electricity prices are from EPEX spot market. We simulate the demand response scheme throughout one day with hourly resolution, i.e., $\tau \in \{1, \ldots, 24\}$, with a random flexibility request ${r}$ signal encompassing both response and rebound blocks.  For small-scale instances, \eqref{eq:reference} can be solved in a centralized manner relying on a big-M reformulation with empirically tuned $M$-constants \cite{bigM}, returning a MIQP. In Figure \ref{fig:agents}, we report the behaviour of prosumer $1$. Results for the remaining agents are similar and omitted in the interest of space. As expected, the agents tend to charge their battery during rebound periods (see, for instance, interval $[1-3]$ h in Figure \ref{fig:agents}) and discharge it during response periods (see, for instance, interval $[3-6]$ in Figure \ref{fig:agents}). As reported in Figure \ref{fig:flexy} (top plot), the community of buildings is able to contribute to the flexibility provision task, with each building contributing proportionally to its battery size. The inability to provide the full upward flexibility request during some time steps, for instance in the interval $[8-9]$ h is due to a limited resource vector $g$ that forces the buildings to first satisfy their internal demand. On the contrary, the inability to meet the downward flexibility request (see, for instance, interval $[1-2]$ h) is a consequence of physical constraints on the size/ramping of the battery. Additionally, Figure \ref{fig:flexy} (bottom plot) suggests that the affine pricing map $h$ implicitly promotes peak shaving \cite{Price2}. During rebound blocks, the overall power taken from the grid, i.e., $\sum_{i=1}^N p_i - k_i$, is higher than the baseline case, akin to a ``valley-filling" behaviour. Finally, we report the distribution share and the pricing map $h$ decided by the \gls{DSO} in Figure \ref{fig:price}. Interestingly, a pattern seems to emerge between $\alpha$ and $h$ with electricity prices lowering considerably during response blocks. This anti-phase behaviour can be explained by considering that a competition between the \gls{DSO}'s objectives might arise: while the flexibility provision objective tends to reduce the amount of power to be purchased, the first term in \eqref{eq:cost_DSO_ori} would indeed benefit from an increase in it, thus the \gls{DSO} tends to offer ``reduced" prices to encourage more buying.

\begin{figure}
    \centering
    \includegraphics[trim=0cm 0cm 0cm 0cm, clip, width=\linewidth]{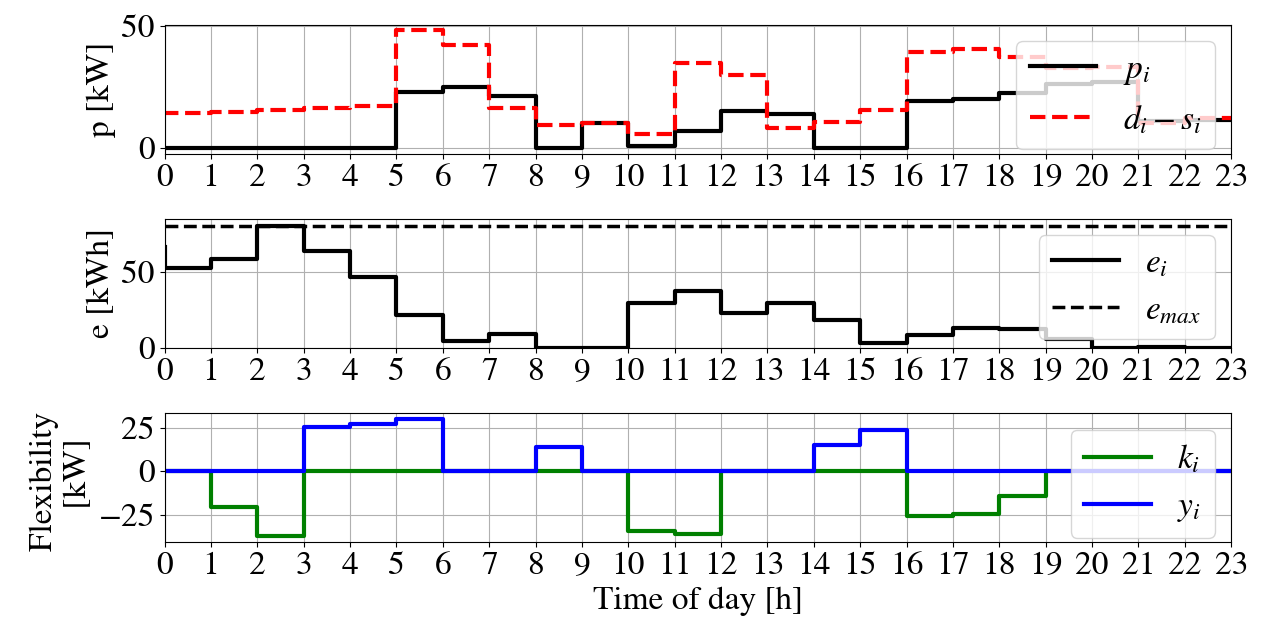}
    \caption{Behaviour of prosumer $i=1$ during the \gls{DR} day-ahead scheduling.}%Top plot represents net demand $d_i - s_i$ (in red) and purchased power from the grid $p_i$ (in black). Middle plot represents the evolution of state of charge of the battery $e_i$. Bottom plot represents the provided response $y_i$ and rebound $k_i$ flexibility.}
    \label{fig:agents}
\end{figure}

\begin{figure}
    \centering
    \includegraphics[trim=0cm 2cm 0cm 0cm, clip, width=\linewidth]{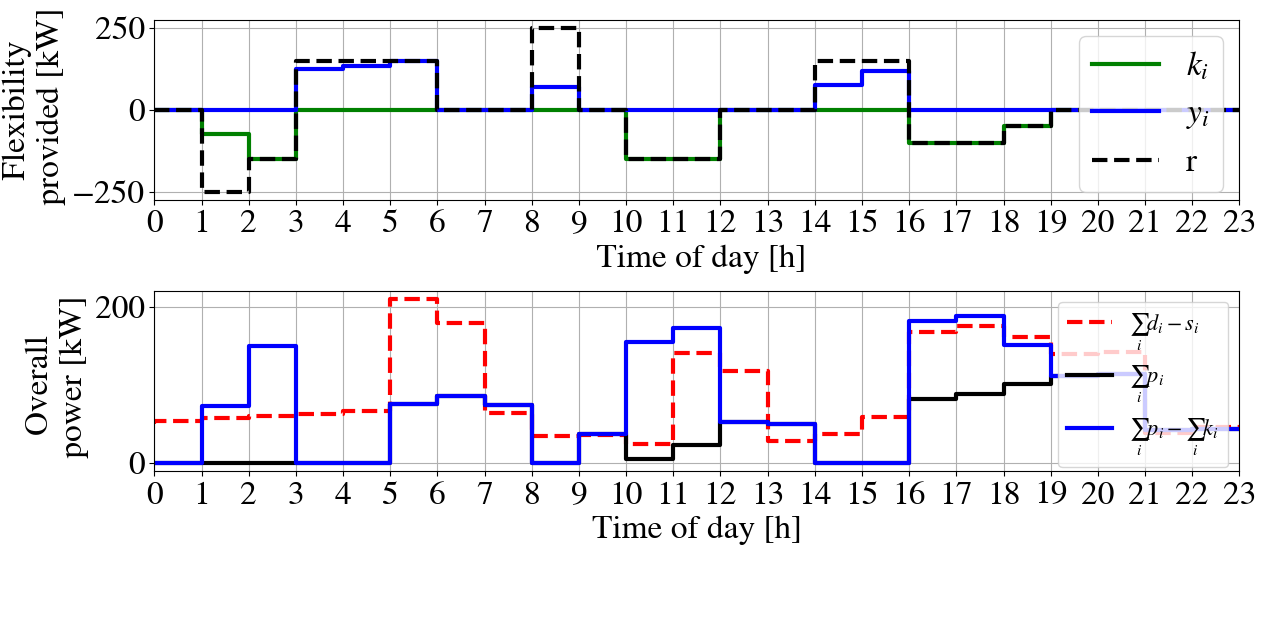}
    \caption{Top plot represents the aggregated power purchased from the grid. Bottom plot represents the aggregated flexibility provision offered by the buildings community.}
    \label{fig:flexy}
\end{figure}

\begin{figure}
    \centering
    \includegraphics[trim=0cm 2cm 0cm 0cm, clip,width=\linewidth]{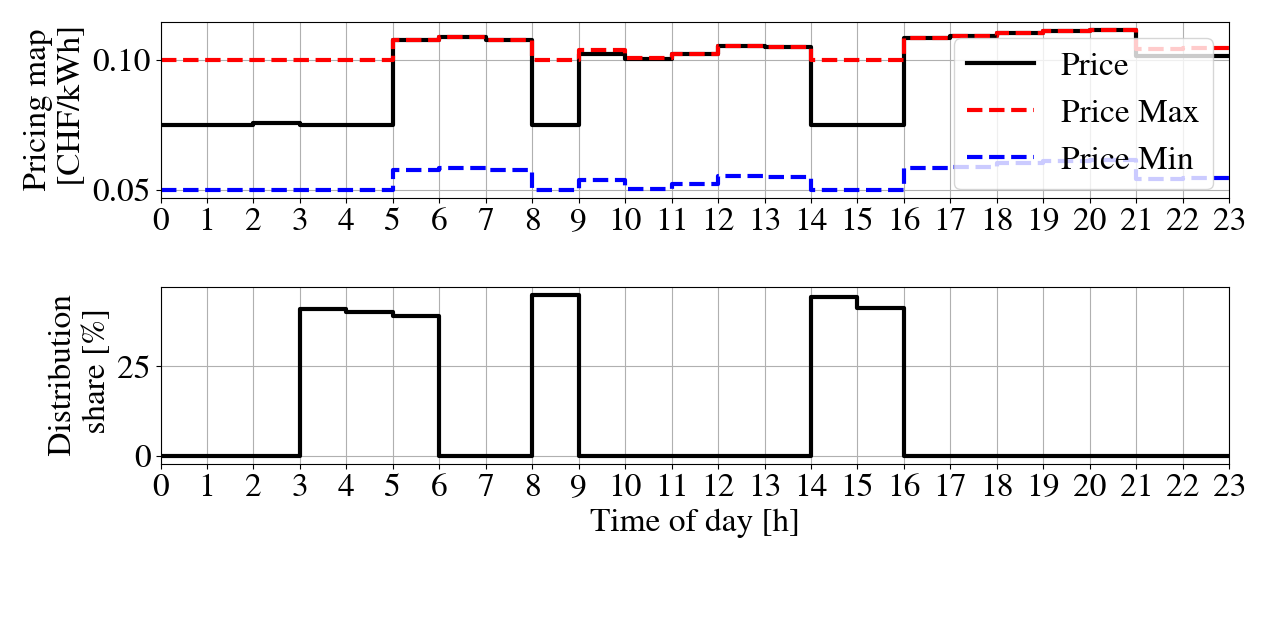}
    \caption{Top plot represents the pricing map $h$. Bottom plot represents the distribution share $\alpha$.}
    \label{fig:price}
\end{figure}

\section{Conclusion} \label{Sec:end}
We present a novel demand response scheme modelled as a Stackelberg game between a \gls{DSO} and a community of prosumers. The scheme is able to couple the local electricity market with ancillary service provision ensuring at the same time fairness. Additionally, it allows a broader action space to both the \gls{DSO} (thanks to the two different revenue streams in place) and the prosumers ({thanks to the usage of saturated reward function}). We prove the existence and convergence to local Stackelberg equilibria for the underlying game. %We show in simulation that the usage of dynamic reward functions with saturation for \gls{DR} is effective in terms of market coupling between wholesale pricing and emergency \gls{DR} provision. 

%Future extensions will compare the presented framework with other design choices and will investigate how to solve the scheme in a distributed way to preserve users privacy and lighten the computational burden. Moreover, we are interested in investigating the eventual emergence of coalitions between the followers and in handling the presence of stochastic uncertainties affecting the system evolution.

\bibliographystyle{IEEEtran}
 \bibliography{biblio_CC.bib}

\end{document}